\documentclass{article}
\usepackage{ijcai22}

\usepackage{times}
\usepackage{soul}
\usepackage{url}
\usepackage[hidelinks]{hyperref}
\usepackage[utf8]{inputenc}
\usepackage[small]{caption}
\usepackage{graphicx}
\usepackage{booktabs}
\usepackage{amsmath,amsthm,amssymb,algpseudocode,graphicx,tikz,caption,subcaption,xcolor,xspace,cleveref,nicefrac,natbib}
\definecolor{light-gray}{gray}{0.7}

\usepackage[shortlabels]{enumitem}

\usepackage[ruled,vlined,linesnumbered]{algorithm2e}

\algrenewcommand\algorithmicindent{0.6em}
\newtheorem{theorem}{Theorem}[section]
\newtheorem{proposition}[theorem]{Proposition}
\newtheorem{lemma}[theorem]{Lemma}
\newtheorem{corollary}[theorem]{Corollary}

\algnotext{EndFor}
\algnotext{EndIf}
\algnotext{EndWhile}
\algnotext{EndProcedure}

\theoremstyle{definition}
\newtheorem{definition}[theorem]{Definition}

\usetikzlibrary{arrows.meta}
\tikzset{>={Latex[width=2mm,length=2mm]}}
\usetikzlibrary{decorations.pathmorphing,shapes,snakes}
\usetikzlibrary{calc}
\usetikzlibrary{arrows, decorations.markings}
\usetikzlibrary{positioning}

\newcommand{\cG}{\mathcal{G}}

\allowdisplaybreaks

\title{Fixing Knockout Tournaments With Seeds}

\author{
Pasin Manurangsi$^1$\And
Warut Suksompong$^2$\\
\affiliations
$^1$Google Research, USA\\
$^2$School of Computing, National University of Singapore, Singapore\\
}

\begin{document}

\maketitle

\begin{abstract}
Knockout tournaments constitute a popular format for organizing sports competitions.
While prior results have shown that it is often possible to manipulate a knockout tournament by fixing the bracket, these results ignore the prevalent aspect of player seeds, which can significantly constrain the chosen bracket.
We show that certain structural conditions that guarantee that a player can win a knockout tournament without seeds are no longer sufficient in light of seed constraints.
On the other hand, we prove that when the pairwise match outcomes are generated randomly, all players are still likely to be knockout winners under the same probability threshold with seeds as without seeds.
In addition, we investigate the complexity of deciding whether a manipulation is possible when seeds are present.
\end{abstract}

\section{Introduction}
\label{sec:intro}

Several practical scenarios involve choosing a winner from a given set of candidates based on pairwise comparisons, perhaps most prominently sports competitions.
A popular format for organizing such competitions is a \emph{knockout tournament}, also known as a \emph{single-elimination} or \emph{sudden death tournament}, wherein the players are matched according to an initial bracket and play proceeds until the tournament winner is determined.
When there are $n$ participating players, a knockout tournament requires arranging only $n-1$ matches, thereby making it an attractive choice among organizers in comparison to a round-robin tournament, for which $\Theta(n^2)$ matches are necessary.
Moreover, the fact that each player is eliminated after a single loss adds a layer of excitement and ensures that no match is meaningless in a knockout tournament. 

As efficient and as exciting as knockout tournaments are, they have a clear drawback in that their winner can depend heavily on the chosen bracket.
A significant line of work in computational social choice has therefore investigated the problem of when it is possible for the organizers to fix a bracket to help their preferred player win the tournament, given the knowledge of which player would win in any pairwise matchup \citep{VuAlSh09,Vassilevskawilliams10,StantonVa11,StantonVa11-2,ChatterjeeIbTk16,RamanujanSz17,AzizGaMa18,GuptaRoSa18,GuptaRoSa18-2,GuptaRoSa19,ManurangsiSu21}.\footnote{For an overview of this line of work, we refer to the surveys by \citet{Vassilevskawilliams16} and \citet{Suksompong21}.}
This problem is known as the \emph{tournament fixing problem (TFP)}.
While TFP is NP-complete \citep{AzizGaMa18}, several structural conditions have been shown to guarantee that a certain player can win a knockout tournament under some bracket, which can be computed efficiently.
For example, \citet{Vassilevskawilliams10} proved that if a player $x$ is a \emph{king}---meaning that for any other player~$y$ who beats $x$, there exists another player $z$ such that $x$ beats $z$ and $z$ beats $y$---and $x$ beats at least $n/2$ other players, then $x$ can win a knockout tournament.
Moreover, a number of papers have shown that fixing knockout tournaments is usually easy when the pairwise match outcomes are drawn from probability distributions \citep{Vassilevskawilliams10,StantonVa11,KimSuVa17,ManurangsiSu21}.

While previous results on TFP have shed light on the manipulability of knockout tournaments, they hinge upon a pivotal assumption that the organizers can choose an arbitrary bracket for their tournament.
In reality, the choice of bracket is often much more constrained.
In particular, many real-world tournaments assign \emph{seeds} to a subset of players in order to prevent highly-rated players from having to play each other too early in the tournament.
For instance, in ATP tennis tournaments with $32$ players, eight players are designated as seeds, and the bracket must be chosen so that the top two seeds cannot meet until the final, the top four seeds cannot meet until the semifinals, and all eight seeds cannot meet until the quarterfinals \citep[p.~139]{ATPTour22}.
As such, algorithms that do not take seeds into account may fail to generate a valid bracket for the competition.
Do prior results in this line of work continue to hold in the presence of seed constraints, or are knockout tournaments more difficult to manipulate in light of these constraints?

\subsection{Our Results}

Following most real-world tournaments, we assume that the knockout tournament is balanced, and that both the number of players $n$ and the number of seeds are powers of two.

We begin in \Cref{sec:structural} by examining structural conditions from the non-seeded setting.
Besides the ``kings who beat at least $n/2$ players'' condition that we already mentioned, another basic condition that suffices for guaranteeing that a player can win a knockout tournament is the ``superking'' condition \citep{Vassilevskawilliams10}.\footnote{See the definition in \Cref{sec:prelim}.}
We show that both of these conditions are no longer sufficient in the seeded setting.
Specifically, for any number of seeds, a king who is not one of the top two seeds may not be able to win a knockout tournament even if it can beat all other players except one.
Likewise, if there are at least four seeds, then a king who is assigned one of the top two seeds and beats all but one player may still be unable to win.
On the positive side, when there are only two seeds, a seeded king who beats at least $n/2+1$ players is guaranteed to be a winner under some bracket, and the bound $n/2 + 1$ is tight.
We also prove similar results for superkings: with two seeds, a superking can always win under some bracket, but as soon as there are at least four seeds, it may no longer be able to win a knockout tournament even if it is one of the top four seeds.
We therefore introduce a stronger condition of ``ultraking'' and show that an ultraking can win a knockout tournament for any number of seeds.

In \Cref{sec:probabilistic}, we investigate the problem of fixing knockout tournaments from a probabilistic perspective.
In particular, we assume that the pairwise outcomes are determined according to the so-called \emph{generalized random model}, where player~$i$ beats player $j$ with probability $p_{ij}$, and these probabilities may vary across different pairs $i,j$ but are always at least a given parameter $p$.
Our prior work has shown that in the non-seeded setting, as long as $p = \Omega(\log n/n)$, all players are likely to be knockout winners \citep{ManurangsiSu21}.
We strengthen that result by showing that the same holds even in the seeded setting, regardless of the number of seeds.
Combined with our findings in \Cref{sec:structural}, this strengthened result shows that even though the presence of seed constraints makes it more difficult to fix tournaments in the worst case, most of the time it does not render manipulation infeasible.

Finally, in \Cref{sec:complexity}, we address the complexity of TFP in the seeded setting.
By reducing from TFP in the non-seeded setting, we show that the problem is NP-complete, both when the number of seeds is an arbitrary constant and when it is $n/2$ (i.e., the highest possible).
On the other hand, we provide an algorithm that solves TFP in $2^n \cdot n^{O(1)}$ time, thereby generalizing a result of \citet{KimVa15} from the setting without seeds.

\section{Preliminaries}
\label{sec:prelim}

As is common in this line of work, we assume that the knockout tournament is balanced and played among $n = 2^r$ players for some positive integer $r$.
In the \emph{seeded setting}, there is a parameter $2\le s\le n$, which we also assume (following the vast majority of real-world tournaments) to be a power of two.
Among the $n$ players, $s$ of them are assigned seeds $1,2,\dots,s$.
The bracket must be set up in such a way that seeds $1$ and $2$ cannot play each other before the final, seeds $1$ through $4$ cannot play each other before the semifinals, seeds $1$ through $8$ cannot play each other before the quarterfinals, and so on.
A bracket satisfying these seed constraints is said to be \emph{valid}.
Observe that $s = n$ is equivalent to $s = n/2$, so we assume henceforth that $2\le s\le n/2$.
Notice also that larger values of $s$ only add extra constraints compared to smaller values of $s$.
We refer to the setting without seeds typically studied in previous work as the \emph{non-seeded setting}.
Unless stated otherwise, our results are for the seeded setting.
The \emph{tournament fixing problem (TFP)} asks whether there exists a valid bracket that makes our desired winner $x$ win the tournament; if the answer is positive, we refer to $x$ as a \emph{knockout winner}.

We are given a \emph{tournament graph} $T=(V,E)$, which indicates the winner of any pairwise matchup.
The vertices in~$V$ correspond to the $n$ players---we will refer to vertices and players interchangeably---and there is a directed edge in $E$ from player~$x$ to player~$y$ whenever $x$ would beat $y$ if they were to play against each other.
We use $x\succ y$ to denote an edge from $x$ to $y$.
We extend this notation to sets of vertices: for $V_1,V_2\subseteq V$, we write $V_1\succ V_2$ to mean that $x\succ y$ for all $x\in V_1$ and $y\in V_2$, and $V_1\succ y$ to mean that $x\succ y$ for all $x\in V_1$.
The \emph{outdegree} of a vertex $x$ is the number of players whom $x$ beats in $T$.
A player $x$ is said to be a \emph{king} if for any other player $y$ who beats $x$, there exists another player~$z$ who loses to $x$ but beats $y$.
Equivalently, $x$ is a king if it can reach every other player via a path of length at most two in $T$.
A player~$x$ is said to be a \emph{superking} if for any other player $y$ who beats~$x$, there exist at least $\log n$ players who lose to $x$ but beat~$y$.\footnote{Throughout the paper, $\log$ refers to the logarithm with base $2$.}
It is clear from the definitions that every superking is also a king.

When constructing a bracket, we will sometimes do so iteratively one round at a time, starting with the first round.
In each round, we pair up the players who remain in that round.
During this process, it can happen that a seeded player is beaten by an unseeded player, or that a stronger-seeded player (i.e., player with a lower-number seed) is beaten by a weaker-seeded player (i.e., player with a higher-number seed).
In such cases, we will think of the (stronger) seed as being ``transferred'' from the losing player to the winning player, since the constraints on the losing player's seed apply to the winning player in subsequent rounds.
We stress that the concept of a transfer is merely for the sake of constructing a bracket and does not imply that the seeds are actually transferred between players when the actual tournament takes place.

\section{Structural Results}
\label{sec:structural}

In this section, we examine the extent to which structural guarantees from the non-seeded setting continue to hold in the seeded setting.
\citet{Vassilevskawilliams10} proved that, without seeds, a king with outdegree at least $n/2$ is a knockout winner,\footnote{We reiterate here that we use the term ``knockout winner'' to refer to a player who can win a knockout tournament \emph{under some bracket}.} and the same holds for a superking.
As we will see, these conditions are largely insufficient for guaranteeing that a player can win in the presence of seeds.

We assume throughout this section that our desired winner~$x$ is a king.
Denote by $A$ the set of players whom $x$ beats, and $B = V\setminus(A\cup\{x\})$ the set of players who beat $x$.
A structure that will be used multiple times in this section is a ``special tournament graph'', defined as follows.

\begin{definition}
\label{def:special}
A tournament graph $T$ is said to be a \emph{special tournament graph} if there exists $y\in A$ such that $y\succ B$ and $B\succ (A\setminus\{y\})$.
The edges within each of $A$ and $B$ can be oriented arbitrarily.
An illustration is shown in \Cref{fig:special-tournament}.
\end{definition}

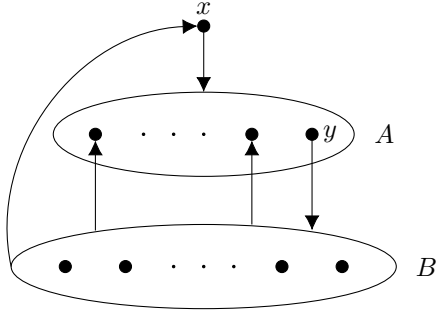
\begin{figure}[!ht]
\centering
\begin{tikzpicture}[scale=0.8]
\draw[fill=black] (6,8) circle  [radius=0.1];
\draw[fill=black] (4.2,6.2) circle  [radius=0.1];
\draw[fill=black] (6.8,6.2) circle  [radius=0.1];
\draw[fill=black] (7.8,6.2) circle  [radius=0.1];
\draw[fill=black] (5,6.2) circle  [radius=0.02];
\draw[fill=black] (5.5,6.2) circle  [radius=0.02];
\draw[fill=black] (6,6.2) circle  [radius=0.02];
\draw (6,6.2) ellipse (2.5cm and 0.7cm);
\draw[fill=black] (3.7,4) circle  [radius=0.1];
\draw[fill=black] (4.7,4) circle  [radius=0.1];
\draw[fill=black] (7.3,4) circle  [radius=0.1];
\draw[fill=black] (8.3,4) circle  [radius=0.1];
\draw[fill=black] (5.5,4) circle  [radius=0.02];
\draw[fill=black] (6,4) circle  [radius=0.02];
\draw[fill=black] (6.5,4) circle  [radius=0.02];
\draw (6,4) ellipse (3.2cm and 0.7cm);
\draw[->] (6,8) to (6,6.9);
\draw[->] (2.8,4) to[bend left=50] (5.9,8);
\draw[->] (7.8,6.2) to (7.8,4.6);
\draw[->] (6.8,4.7) to (6.8,6.1);
\draw[->] (4.2,4.6) to (4.2,6.1);
\node at (9,6.2) {$A$};
\node at (9.7,4) {$B$};
\node at (6,8.3) {$x$};
\node at (8.1,6.2) {$y$};
\end{tikzpicture}
\caption{Illustration of a special tournament graph}
\label{fig:special-tournament}
\end{figure}

When we refer to a special tournament graph, we will use the notation $x,y,A,B$ as in \Cref{fig:special-tournament}.

The results of this section are summarized in \Cref{table:summary}.

\renewcommand{\arraystretch}{1.3}
\begin{table*}[!ht]
\centering
    \begin{tabular}{| c | c |}
    \hline
     \textbf{Conditions} & \textbf{Knockout winner guarantee} \\ \hline
     Any $n\ge 4$, any $s$, not one of the top two seeds, outdegree $n-2$ & No (\Cref{prop:king-not-seeded}) \\ \hline
     $s=2$, any $n$, seeded, outdegree $ \ge n/2 + 1$ & Yes (\Cref{thm:king-seeded-positive})  \\ \hline
     $s=2$, any $n\ge 4$, seeded, outdegree $n/2$ & No (\Cref{thm:king-seeded-negative})   \\ \hline
     Any $4\le s\le n/2$, one of the top two seeds, outdegree $n-2$ & No (\Cref{prop:king-seeded-negative}) \\ \hline
     $s=2$, any $n$, superking & Yes (\Cref{thm:superking-positive}) \\ \hline
     Any $4\le s\le n/2$, one of the top four seeds, superking & No (\Cref{prop:superking-negative}) \\ \hline
     Any $n$ and $s$, ultraking & Yes (\Cref{thm:ultraking-positive}) \\
    \hline
    \end{tabular}
    \caption{Summary of our results in \Cref{sec:structural} on whether each set of conditions is sufficient to guarantee that a king can win a knockout tournament under some bracket.}
    \label{table:summary}
\end{table*}

\subsection{Kings of High Outdegree}

We begin by observing that in the seeded setting, a king~$x$ may not be able to win a knockout tournament even if it beats $n-2$ other players.
This draws a sharp contrast to the non-seeded setting, where beating $n/2$ other players already suffices \citep{Vassilevskawilliams10}.
Note that the bound $n-2$ is also tight: if $x$ beats $n-1$ players, then it beats every player and can trivially win.

\begin{proposition}
\label{prop:king-not-seeded}
For any $n\ge 4$ and any $s$, a king with outdegree $n-2$ who is not one of the top two seeds is not necessarily a knockout winner.
\end{proposition}

\begin{proof}
Consider a special tournament graph with $|A| = n-2$, and let $z$ be the unique player in $B$.
Assume that $y$ and $z$ are the top two seeds.
Since $y$ is the only player who can beat $z$, and these two players cannot meet until the final, $z$ makes the final in every valid bracket.
Thus, even if the king $x$ makes the final, it will be beaten by $z$ there, which means that $x$ is not a knockout winner.
\end{proof}

The construction in the proof above relies on the assumption that $x$ is not one of the top two seeds.
We show next that if there are only two seeds and $x$ is one of them, then it can win a knockout tournament as long as it has outdegree at least $n/2 + 1$.
This means that the seed constraint does not make manipulation much harder in this case.
To establish this result, we employ a similar algorithm as that of \citet{Vassilevskawilliams10} for the ``kings who beat at least $n/2$ players'' guarantee, but our analysis is more involved due to the seed constraint.

\begin{theorem}
\label{thm:king-seeded-positive}
For $s = 2$ and any $n$, a seeded king $x$ with outdegree at least $n/2 + 1$ is always a knockout winner.
Moreover, there exists a polynomial-time algorithm that computes a valid winning bracket for $x$.
\end{theorem}

\begin{proof}
The case $n=2$ holds vacuously, so assume that $n\ge 4$.
Since $s = 2$, the only constraint is that the two seeds cannot meet before the final.
We will simultaneously prove the following two statements by induction on $n$:
\begin{enumerate}[(a)]
\item A seeded king $x$ with outdegree at least $n/2 + 1$ is always a knockout winner.
\item For a special tournament graph with seeded king $x$ and $|A| = n/2$, if the other seed belongs to $B\cup\{y\}$, then $x$ is a knockout winner.
\end{enumerate}
Note that even though only statement (a) is needed for this theorem, a proof by induction using statement (a) alone does not work, and we also need statement (b) in order for the induction to go through.

We first handle the base case $n = 4$.
Statement (a) holds trivially because $x$ beats all other players.
For statement (b), denote by $z$ the unique player in $B$ and by $w$ the player in~$A$ besides $y$.
By the assumption of the statement, $w$ is not a seed.
We let $x$ play $w$ and $y$ play $z$ in the first round, so that $x$ beats $y$ in the final and wins the tournament.

We proceed to the inductive step.
Assume that the statements hold for $n/2\ge 4$; we will prove both of them for~$n$ at the same time.
Consider the following algorithm.
In the first round, we find a maximum matching $M$ from $A$ to $B$ (in the underlying directed graph between $A$ and $B$) and pair up players according to $M$.
Note that the matching $M$ is always nonempty, and let $A_M$ and $B_M$ be the players from $A$ and $B$ in the matching, respectively.
Among the remaining players, we match $x$ with an unseeded player in $A$, match players within $A$ arbitrarily, match players within $B$ arbitrarily, and finally, if necessary, match the leftover player in $B$ with the leftover player in $A$.
Denote by $A'\subseteq A$ and $B'\subseteq B$ the players in $A$ and $B$ who remain after this round, respectively.
We consider three cases.
\begin{itemize}
\item \underline{Case 1}: $|A| \ge n/2 + 2$.
We have $|B| \le n/2 - 3$, so after matching players according to $M$, there are still at least five players left in $A$.
This means that we can match $x$ with an unseeded player in $A$.
Since $M$ has size at least~$1$ and we match one player in $A$ with $x$, $|A'| \ge n/4 + 1$.
Moreover, every player $z\in B\setminus B_M$ beats every player in $A\setminus A_M$ (otherwise the matching can be enlarged), so since $x$ is a king, $z$ must lose to at least one player in $A_M$.
This implies that $x$ is still a king in the remainder tournament.
By the inductive hypothesis for statement~(a), there exists a winning bracket for $x$ in the remainder tournament, and therefore $x$ can also win the original tournament.

\item \underline{Case 2}: $|A| = n/2 + 1$.
We have $|B| = n/2 - 2$, so by a similar reasoning as in Case~1, we can match $x$ with an unseeded player in $A$.
If $M$ has size at least $2$, then $|A'| \ge n/4 + 1$ and we are done as in Case~1.
Assume therefore that $M$ has size $1$, which means that the tournament graph is a special tournament graph, with a player $y\in A$ beating all players in $B$.
In this special case, we add an extra constraint to the algorithm: if the seed besides $x$ belongs to $A\setminus \{y\}$, then we leave this seed to be paired with a player in $B$ in the final step of the algorithm, so that the seed is transferred to $B$ for the next round (since $|A|$ is odd, the final step of the algorithm takes place).
This ensures that in the remainder tournament, which also has a special tournament graph, the seed other than $x$ belongs to $B'\cup\{y\}$.
Since $|A'| = n/4$, we are done by the inductive hypothesis for statement~(b).

\item \underline{Case 3}: $|A| = n/2$, the tournament graph is a special tournament graph, and the seed besides $x$ belongs to $B\cup\{y\}$.
In this case, the maximum matching~$M$ has size $1$.
After applying the algorithm for the first round, the remainder tournament still has a special tournament graph, $|A'| = n/4$, and the seed other than $x$ belongs to $B'\cup\{y\}$.
Hence, we are done by the inductive hypothesis for statement~(b).
\end{itemize}

The three cases together complete the induction.
Since constructing each round of the bracket takes polynomial time and the number of rounds is logarithmic, the overall algorithm runs in polynomial time.
\end{proof}

Interestingly, our next result establishes the tightness of the bound $n/2 + 1$ in \Cref{thm:king-seeded-positive}---this provides a separation from the $n/2$ bound in the non-seeded setting.
To this end, we will need the following lemma.
The case $|A| = n/2 - 1$ of the lemma was proven as Claim~2 in the work of \citet{Vassilevskawilliams10}, and the same proof applies to the more general condition $|A| \le n/2-1$ that we state below.

\begin{lemma}[\citep{Vassilevskawilliams10}]
\label{lem:special-negative}
In the non-seeded setting, for a special tournament graph with king $x$ and $|A| \le n/2 - 1$, $x$ is not a knockout winner.
\end{lemma}

\begin{theorem}
\label{thm:king-seeded-negative}
For $s = 2$ and any $n\ge 4$, a seeded king with outdegree $n/2$ is not necessarily a knockout winner.
\end{theorem}

\begin{proof}
Consider a special tournament graph with $|A| = n/2$, and assume that the other seed $w\ne x$ belongs to $A\setminus \{y\}$.
We prove by induction on $n$ that $x$ is not a knockout winner for this tournament graph.
For the base case $n = 4$, denote by $z$ the unique player in $B$.
Since $x$ cannot play $w$ in the first round, it must play $y$ in order to progress to the final.
However, since $z$ beats $w$, $x$ will play $z$ in the final and lose.

We proceed to the inductive step.
Assume that the statement holds for $n/2\ge 4$; we will prove it for $n$.
In order for $x$ to have a chance of winning the tournament, we must match it with a player from $A$.
This player must be different from~$y$; otherwise no player outside $B$ can eliminate players in $B$.
If we match $y$ with a player in $B$ and pair up the remaining $n/2 - 2$ players in $B$, then $n/4 - 1$ players from~$B$ proceed to the next round, the tournament graph is still a special tournament graph, and the seeds are still $x$ and a player in $A\setminus\{y\}$, so we may apply the inductive hypothesis.
Otherwise, at least $n/4$ players from $B$ proceed to the next round.
If $y$ does not proceed, no player outside $B$ can eliminate players in $B$ and $x$ cannot win, so we may assume that $y$ proceeds.
In this case, the tournament graph is again a special tournament graph.
Since $|B| \ge n/4$, we have $|A| \le n/4 - 1$, so \Cref{lem:special-negative} implies that $x$ is not a knockout winner.
This completes the induction and therefore the proof.
\end{proof}

When $s \ge 4$, we know from \Cref{prop:king-not-seeded} that if the king~$x$ is not one of the top two seeds, then it may not be able to win even if its outdegree is $n-2$.
A similar construction shows that the same also holds when $x$ \emph{is} one of the top two seeds, which means that manipulation with $s \ge 4$ is more difficult than with $s = 2$.

\begin{proposition}
\label{prop:king-seeded-negative}
For any $4\le s\le n/2$, a king $x$ with outdegree $n-2$ is not necessarily a knockout winner even if $x$ is one of the top two seeds.
\end{proposition}

\begin{proof}
We use the construction in the proof of \Cref{prop:king-not-seeded}, but with the extra specification that $y$ loses to all other players in $A$.
The top two seeds are $x$ and $z$, and $y$ is the third seed.
Since $n\ge 8$, $y$ cannot be matched with $z$ in the first round, so it gets eliminated as it cannot beat any other player.
But then no player can eliminate $z$, which means that $x$ cannot win the tournament.
\end{proof}

\subsection{Superkings}

Next, we consider the superking condition---recall that a superking is always a knockout winner in the non-seeded setting \citep{Vassilevskawilliams10}.
We prove that if there are only two seeds, then a superking can win the tournament regardless of whether it is seeded.
This stands in contrast to \Cref{prop:king-not-seeded}, which shows that a king may not be able to win for any $s$ even if it has outdegree $n-2$.

\begin{theorem}
\label{thm:superking-positive}
For $s = 2$ and any $n$, a superking $x$ is always a knockout winner.
Moreover, there exists a polynomial-time algorithm that computes a valid winning bracket for $x$.
\end{theorem}

\begin{proof}
The case $n = 2$ is trivial, so assume that $n\ge 4$.
First, recall the superking algorithm of \citet{Vassilevskawilliams10} in the non-seeded setting: Match the superking $x$ with an arbitrary player $w\in A$, find a maximum matching $M$ from $A\setminus\{w\}$ to $B$ and pair up players according to $M$, and match the remaining players arbitrarily.
It can be shown that these first-round matchups ensure that $x$ remains a superking in the next round, so we can proceed recursively.

To demonstrate that this algorithm can be applied in the seeded setting, it suffices to show that in each round before the final, we can avoid pairing the two seeds.\footnote{As mentioned in \Cref{sec:prelim}, seeds can be ``transferred'' as the tournament proceeds; this does not affect our argument.}
If $x$ is one of the seeds, then since the superking condition implies that $|A| \ge 2$ in every round before the final, we can match $x$ with an unseeded player $w\in A$ in the first step of the algorithm.
Assume now that $x$ is not one of the seeds.
If at least one seed is in $A$, we can let $x$ play against that seed.
Suppose therefore that both seeds are in $B$.
If at least one of them is used in $M$, we are done.
Otherwise, we try to avoid pairing the two seeds when matching the remaining players.
The only problematic case is when the maximum matching has already exhausted $A$ and left exactly the two seeds in $B$.
In this case, since each of the two seeds loses to at least $\log n \ge 2$ players in $A$ by the superking condition, we can modify one pair in $M$ to involve one of the seeds, and we are again done.

It is clear that this algorithm can be implemented in polynomial time.
\end{proof}

On the other hand, when there are at least four seeds, the superking condition is no longer sufficient.

\begin{proposition}
\label{prop:superking-negative}
For any $4\le s\le n/2$, a superking is not necessarily a knockout winner even if it is one of the top four seeds.
\end{proposition}

\begin{proof}
Consider a tournament graph in which the superking~$x$ beats exactly $\log n$ players in $A$, all of whom in turn beat the remaining $n-1-\log n$ players in $B$.
Notice that $|A| = \log n \ge 3$, and assume that $x$ and three players in $A$ are the top four seeds.
In the first $\log n - 2$ rounds of the tournament (i.e., before the semifinals), $x$ cannot play these three seeds, so it can only play the other $\log n - 3$ players whom it beats.
As a result, $x$ cannot progress to the semifinals.
\end{proof}

In light of \Cref{prop:superking-negative}, we introduce the following strengthening of a superking, where we replace the parameter $\log n$ in its definition by $n/2$.

\begin{definition}
\label{def:ultraking}
A player $x$ is said to be an \emph{ultraking} if for any other player $y$ who beats $x$, there exist at least $n/2$ players who lose to $x$ but beat $y$.
\end{definition}

Our next theorem shows that the ultraking condition is sufficient to guarantee that a player can win a knockout tournament in the seeded setting, regardless of the number of seeds.

\begin{theorem}
\label{thm:ultraking-positive}
For any $n$ and $s$, an ultraking $x$ is always a knockout winner.
Moreover, there exists a polynomial-time algorithm that computes a valid winning bracket for $x$.
\end{theorem}

\begin{proof}
The case $n=2$ is trivial, so assume that $n\ge 4$.
Since larger values of $s$ only add extra constraints, it suffices to consider $s = n/2$, i.e., half of the players are seeded and the other half are unseeded.

Denote by $k$ the number of seeded players in $B$; there are at most $n/2 - k$ seeded players in $A$.
Since each player in~$B$ loses to at least $n/2$ players in $A$, it loses to at least $k$ unseeded players in $A$.
This means that we can greedily match each seeded player in $B$ with an unseeded player in $A$ to whom it loses.
Analogously, we can also match each unseeded player in $B$ with a seeded player in~$A$ to whom it loses.
The remaining matches (between players in $A$ and the ultraking $x$) can be chosen arbitrarily so that each seeded player is matched to an unseeded player---this is possible since there are an equal number of seeded and unseeded players.
This ensures that after the first round, only $x$ and players in $A$ remain.
Hence, no matter which (valid) bracket we choose from the second round onward, $x$ is the tournament winner.

It is clear that this algorithm can be implemented in polynomial time.
\end{proof}

We remark that \Cref{thm:ultraking-positive} would no longer hold if the parameter $n/2$ in \Cref{def:ultraking} were reduced to $n/2 - 1$.
Indeed, if $s = n/2$ and the ultraking $x$ is seeded and beats the other $n/2 - 1$ seeded players in $A$, all of whom beat the $n/2$ unseeded players in $B$, then $x$ cannot even win its first-round match in any valid bracket.

\section{Probabilistic Results}
\label{sec:probabilistic}

In this section, we investigate the problem of fixing tournaments from a probabilistic perspective.
We work with the so-called \emph{generalized random model} \citep{KimSuVa17,SaileSu20}. 
In this model, we are given a parameter $p \in [0, 1/2]$, and every pair of distinct $i,j\in\{1,\dots,n\}$ is assigned a real number $p_{ij} \in [p, 1 - p]$, where $p_{ij} = 1 - p_{ji}$ for all $i\ne j$.
Assuming that the players are $x_1,\dots,x_n$, the tournament graph~$T$ is then generated as follows: for each pair $i \ne j$, player~$x_i$ beats player~$x_j$ with probability $p_{ij}$.

Our main result of this section is stated below.
Following standard terminology, ``with high probability'' means that the probability converges to $1$ as $n\rightarrow\infty$.

\begin{theorem} \label{thm:random-main}
Let $s = n/2$ and $p \geq 160\ln n/n$.
With high probability, all players are knockout winners, and a winning bracket for each player can be computed in polynomial time.
\end{theorem}

As mentioned in \Cref{sec:prelim}, $s = n/2$ gives rise to the most restrictive set of constraints, so \Cref{thm:random-main} implies the same result for all other values of $s$.
The theorem also strengthens our previous result in the non-seeded setting \citep{ManurangsiSu21}.\footnote{Modulo the constant factor, which is twice as large as our previous one.}
Moreover, it entails similar results for random models that are special cases of the generalized random model, including the ``Condorcet random model'' and the ``uniform random model''.\footnote{See, e.g., our prior work \citep{ManurangsiSu21} for the definitions. We remark here that the parameter $p$ is present in the Condorcet random model but not in the uniform random model.}
The bound $p = \Omega(\log n/n)$ is tight even in the non-seeded setting: if $p = o(\log n/n)$ and $p_{ij} = p$ for all $i > j$, then the weakest player is expected to beat fewer than $\log n$ players, which is insufficient because a knockout tournament consists of $\log n$ rounds.

Before we proceed to the proof of \Cref{thm:random-main}, we provide here a brief sketch.
For ease of exposition, we assume that our desired winner $x$ is unseeded. 
To begin with, we use our prior result from the non-seeded setting \citep{ManurangsiSu21} to find a winning bracket~$B$ for $x$ among the $n/2$ unseeded players. 
We then extend this bracket into a full bracket of size $n$ so that after the first round is played, the bracket becomes $B$. 
To achieve this, we assign seeds $1$ and $2$ to the two halves in the first step, and then seeds $2^{k-1} + 1, \dots, 2^{k}$ in each step $k \geq 2$. 
For each $2\le k\le \log n - 1$, we have the freedom of assigning the $2^{k - 1}$ seeds to any ``section'' of the bracket containing $2^{n-k}$ players with the property that the section has not been assigned any seeded player thus far (there are $2^k - 2^{k-1} = 2^{k - 1}$ such sections).
Since we want the bracket to become $B$ after the first round, we must also ensure that each seeded player loses in the first round. 
To this end, we create a bipartite graph between the seeds $2^{k-1} + 1, \dots, 2^{k}$ and the unassigned sections, where an edge is present if the seed loses against at least one unseeded player in that section. 
We then find a perfect matching in this graph and assign these seeds accordingly. 
The main technical aspect of the proof lies in showing that such a perfect matching is likely to exist. 
For this, we need to extend the classic result of~\cite{ErdosRe64} on the existence of perfect matchings in random graphs to a larger parameter regime.  

\subsection{Preliminaries}

\subsubsection{Existence of Perfect Matching in Random Graphs}

Recall that the \emph{Erd{\H{o}}s-R{\'{e}}nyi random bipartite graph distribution} $\cG(m, m, q)$ (where $m$ is a positive integer and $q$ a real number in $[0, 1]$) is the distribution of balanced bipartite graphs with $m$ vertices on each side such that for each pair of left and right vertices, there is an edge between them with probability~$q$ independently of other pairs. 

Our proof requires the high-probability existence of a perfect matching in $\cG(m, m, q)$, stated below. This statement is a slight generalization of the original result by~\citet{ErdosRe64}; our version provides a sharper bound in the regime where $q = 1 - \delta$ is close to $1$ compared to Erd\H{o}s and R\'{e}nyi's version. 
This sharpened bound will be needed for our proof to go through.

\begin{proposition}
\label{prop:er-matching}
Let $G$ be a bipartite graph sampled from the Erd{\H{o}}s-R{\'{e}}nyi random bipartite graph distribution $\mathcal{G}(m, m, 1 - \delta)$. If $\delta^{m/8} \leq 1/m$, then $G$ contains a perfect matching with probability at least $1 - \delta^{m/4}$.
\end{proposition}

Following previous work (e.g.,~\citep{SudakovVu08,ManurangsiSu20}), we prove \Cref{prop:er-matching} by bounding the probability that the graph fails to satisfy the condition of Hall's Marriage Theorem, which we recall next. 
Denote by $N_G(S)$ the set of vertices adjacent to at least one vertex in $S$.

\begin{proposition}[Hall's Marriage Theorem] \label{prop:hall}
Let $G = (A, B, E)$ be any bipartite graph such that $|A| = |B|$. If $|N_G(S)| \geq |S|$ for every subset $S \subseteq A$, then $G$ has a perfect matching.
\end{proposition}

\begin{proof}[Proof of \Cref{prop:er-matching}]
Let $G = (A, B, E)$ be a graph drawn from $\cG(m, m, q)$. 
From \Cref{prop:hall}, the probability that it does \emph{not} contain a perfect matching can be written as
\begin{align*}
&\Pr[\exists S \subseteq A, |N_G(S)| < |S|] \\
&= \Pr[\exists S \subseteq A, T \subseteq B, |T| = |S| - 1, N_G(S) \subseteq T] \\
&\leq \sum_{S \subseteq A} \sum_{T \subseteq B \atop |T| = |S| - 1} \Pr[N_G(S) \subseteq T].
\end{align*}
where the inequality follows from the union bound.

Next, observe that $N_G(S) \subseteq T$ if and only if there is no edge from $S$ to any of the vertices in $B \setminus T$. 
The latter happens with probability $\delta^{|S| \cdot |B \setminus T|} = \delta^{|S|(m + 1 - |S|)}$. 
Plugging this back into the bound above, we have that the probability that a perfect matching does not exist is at most
\begin{align*}
&\sum_{S \subseteq A} \sum_{T \subseteq B \atop |T| = |S| - 1} \delta^{|S|(m + 1 - |S|)} \\
&= \sum_{i=1}^m \sum_{S \subseteq A \atop |S| = i} \sum_{T \subseteq B \atop |T| = i - 1} \delta^{i(m + 1 - i)} \\
&= \sum_{i=1}^m \binom{m}{i} \binom{m}{i - 1} \cdot \delta^{i(m + 1 - i)} \\
&\leq \sum_{i=1}^m m^{\min\{i, m - i\}} \cdot m^{\min\{i - 1, m + 1 - i\}} \cdot \delta^{i(m + 1 - i)} \\
&\leq \sum_{i=1}^m m^{2\cdot\min\{i, m + 1 - i\}} \cdot \delta^{i(m + 1 - i)} \\
&= \sum_{i=1}^m m^{2\cdot \min\{i, m + 1 - i\}}\cdot  \delta^{\min\{i, m + 1 - i\} \cdot \max\{i, m + 1 - i\}} \\
&= \sum_{i=1}^m \left(m^2 \cdot \delta^{\max\{i, m + 1 - i\}}\right)^{\min\{i, m + 1 - i\}}.
\end{align*}
Now, from our assumption on $\delta$, we have $m^2 \leq (1/\delta)^{m/4} \leq (1/\delta)^{\max\{i, m + 1 - i\}/2}$. Therefore, the above term is at most
\begin{align*}
\sum_{i=1}^m \delta^{\max\{i, m + 1 - i\} \cdot \min\{i, m + 1 - i\} / 2} 
&= \sum_{i=1}^m \delta^{i(m + 1 - i)/2} \\
&\leq m \cdot \delta^{m/2} \\
&\leq \delta^{m/4},
\end{align*}
where the last inequality once again follows from our assumption on $\delta$.
\end{proof}

In our proof of \Cref{thm:random-main}, we will need to apply \Cref{prop:er-matching} repeatedly to a specific setting of parameters. 
As such, it will be more convenient to state the following instantiation of the value of $\delta$ in \Cref{prop:er-matching}.

\begin{corollary}
\label{cor:er-matching-customized}
Let $G$ be a bipartite graph sampled from the Erd{\H{o}}s-R{\'{e}}nyi random bipartite graph distribution $\mathcal{G}(m, m, 1 - \delta)$. Let $n \geq m$ be a positive integer and let $p \in [0, 1]$ be such that $p \geq 128\ln n/n$. If $\delta \leq (1 - p)^{n/(16m)}$, then $G$ contains a perfect matching with probability at least $1 - 1/n^2$.
\end{corollary}

\begin{proof}
First, note that
\begin{align*}
\delta^{m/8} \leq (1 - p)^{n/128} \leq e^{-pn/128} \leq 1/n \leq 1/m,
\end{align*}
where for the second inequality we use the well-known estimate $1+x\le e^x$, which holds for all real numbers $x$.
Therefore, we may apply \Cref{prop:er-matching}, which implies that $G$ contains a perfect matching with probability at least
\begin{align*}
1 - \delta^{m/4} \geq 1 - (1 - p)^{n/64} \geq 1 - e^{-pn/64} \geq 1 - 1/n^2,
\end{align*}
as desired.
\end{proof}

\subsubsection{Fixing Random Tournaments in the Non-Seeded Setting}

As mentioned in the proof sketch, we will also use our prior result from the non-seeded setting, which is stated below.

\begin{lemma}[\citep{ManurangsiSu21}] \label{thm:non-seed-ms}
In the non-seeded setting, if $p \geq 80\ln n / n$, then for each player, with probability $1 - o(1/n)$, the player is a knockout winner and a winning bracket for the player can be found in polynomial time.\footnote{The bound $1-o(1/n)$ was derived in the proof of this result; see Theorem~4.3 in the extended version of our previous paper \citep{ManurangsiSu21}.}
\end{lemma}

To prove \Cref{thm:random-main} when the desired winner is unseeded, it suffices to apply \Cref{thm:non-seed-ms} among the unseeded players as the first step. However, for the case where the desired winner is seeded, we need a slightly modified version of \Cref{thm:non-seed-ms} that considers a tournament with $n + 1$ players instead and picks out one player $y$ who loses to the desired winner $x$ (to be played with $x$ in the first round in the proof of \Cref{thm:random-main}). This modified lemma is stated below.
Its proof is nearly identical to that of~\Cref{thm:non-seed-ms}, so we only provide a proof sketch here.

\begin{lemma} \label{lem:non-seed-mod}
Suppose that we consider a tournament graph on $n + 1$ players generated under the generalized random model with $p \geq 80\ln n / n$. Let $x$ be any player. Then, with probability $1 - o(1/n)$, there exists a player $y$ with $x\succ y$ such that $x$ is a knockout winner in the $n$-player (non-seeded) tournament that results from removing $y$. Furthermore, $y$ and a winning bracket for $x$ can be found in polynomial time. 
\end{lemma}

\begin{proof}[Proof Sketch]
The proof follows the same outline as our previous one \citep{ManurangsiSu21}. Specifically, let $Z$ denote the set of players with expected outdegree at least $0.45n + 2$ before the tournament graph is generated (as opposed to $0.45n + 1$ in our previous proof). Our previous argument implies that with probability $1 - o(1/n)$, $x$ beats at least $\log n + 1$ players in $Z$. Let $z_0, \dots, z_{\log n}$ be such players. Set $y = z_0$ and consider the rest of the tournament. We may then follow our previous proof, which shows that with probability $1 - o(1/n)$, there exists a bracket such that $x$ gets to play $z_i$ in round $i$, therefore ensuring that $x$ wins the knockout tournament.
\end{proof}

\subsection{Main Algorithm: Proof of \Cref{thm:random-main}}

We now have all the ingredients to prove \Cref{thm:random-main}. 
For $0\le k\le \log n$, we define a \emph{level-$k$ section} of a bracket of size~$n$ to consist of the positions $(w - 1) \cdot 2^{\log n-k} + 1, \dots, w \cdot 2^{\log n - k}$ in the bracket for some $w \in \{1, \dots, 2^k\}$. For example, a level-$0$ section refers to the entire bracket, a level-$1$ section refers to half of the bracket, and so on.  
Recall that there are $n/2$ seeds.

\begin{proof}[Proof of \Cref{thm:random-main}]
Consider any player $x$. We will show that with probability $1 - o(1/n)$, we can efficiently find a winning bracket for $x$. Taking a union bound over all players~$x$ establishes the theorem.

We consider two cases based on whether $x$ is seeded.

\paragraph{Case I: $x$ is unseeded.}
Our algorithm works as follows:
\begin{itemize}
\item First, use the algorithm from \Cref{thm:non-seed-ms} to find a winning bracket $B$ for $x$ in the tournament of size $n/2$ consisting of all unseeded players. If such a bracket cannot be found, fail.
\item We will extend the bracket $B$ to the entire tournament of size $n$ as follows. Start with a bracket $B'$ of size $n$ where position $i$ is empty if $i$ is odd and is assigned the player in position $i/2$ of $B$ if $i$ is even.
\item For $k = 1, \dots, \log n - 1$:
\begin{itemize}
\item If $k = 1$, let $T_k$ consist of seeds $1$ and $2$.  Otherwise, let $T_k$ consist of seeds $2^{k - 1} + 1, \dots, 2^k$.
\item Let $B^k_1, \dots, B^k_{|T_k|}$ denote the level-$k$ sections of $B'$ such that no player with seed number smaller than $2^k$ has been assigned. 
\item Let $G_k = (T_k, \{B^k_1, \dots, B^k_{|T_k|}\}, E)$ be the bipartite graph such that there is an edge between $u \in T_k$ and $B^k_i$ if and only if $u$ loses to an unseeded player in $B^k_i$.
\item Find a perfect matching in $G_k$. If no perfect matching exists, fail. Otherwise, if $u \in T_k$ is matched to $B^k_i$ in the perfect matching, then let $u$ play an unseeded player it loses to in $B^k_i$ in the first round.
\end{itemize}
\end{itemize}
Clearly, the algorithm runs in polynomial time and, if it succeeds, $x$ is the winner. Therefore, it suffices to show that the algorithm fails with probability $o(1/n)$. 

Consider $G_k$. 
Notice that there are $n/2^{k + 1}$ unseeded players in each $B^k_i$. 
Thus, for any $u \in T_k$, the probability that it loses to at least one player in $B^k_i$ is at least $1 - (1 - p)^{n/2^{k + 1}}$. 
Therefore, we may apply\footnote{Note that while the probability that each edge in $G_k$ exists is at least (and may not be exactly) $1 - \delta$, we can still apply \Cref{cor:er-matching-customized} because increasing the probability of an edge only increases the probability that a perfect matching exists.} \Cref{cor:er-matching-customized} with $\delta = (1-p)^{n/2^{k+1}}$ and $m = |T_k| \geq 2^{k-1}$, which ensures that a matching exists with probability at least $1 - 1/n^2$.

Taking a union bound over all $k$, the probability that all of $G_1, \dots, G_{\log n-1}$ admit a perfect matching is at least $1 - \log n / n^2 = 1 - o(1/n)$. 
Moreover, since $p \ge 160\ln n/n \ge 80\ln (n/2)/(n/2)$, \Cref{thm:non-seed-ms} ensures that the first step of the algorithm fails with probability at most $o(1/n)$. Therefore, the entire algorithm succeeds with probability at least $1 - o(1/n)$, as desired.

\paragraph{Case II: $x$ is seeded.} This case is very similar to Case~I except that we need to use \Cref{lem:non-seed-mod} in order to ensure that we can pair $x$ with $y$ in the first round, and we need to be slightly more careful in the subsequent steps as $x$ and $y$ are already matched. More formally, our algorithm in this case works as follows:
\begin{itemize}
\item First, use the algorithm from \Cref{lem:non-seed-mod} to find an unseeded player~$y$ with $x\succ y$ and a winning bracket $B$ for $x$ in the tournament of size $n/2 + 1$ consisting of all unseeded players and $x$. 
If such a bracket cannot be found, fail.
\item We will extend the bracket $B$ to the entire tournament of size $n$ as follows. 
Start with a bracket $B'$ of size $n$ where position $i$ is empty if $i$ is odd and is assigned the player in position $i/2$ of $B$ if $i$ is even. 
Furthermore, let $x$ play $y$ in the first round of $B'$.
\item For $k = 1, \dots, \log n - 1$:
\begin{itemize}
\item If $k = 1$, let $T_k$ consist of seeds $1$ and $2$. Otherwise, let $T_k$ consist of seeds $2^{k-1} + 1, \dots, 2^k$. Then, let $T'_k = T_k \setminus \{x\}$.
\item Let $B^k_1, \dots, B^k_{|T'_k|}$ denote the level-$k$ sections of $B'$ such that no player with seed number smaller than $2^k$ has been assigned. 
For each $i = 1,\dots, |T'_k|$, if $x$ is not in $B^k_i$, let $U^k_i$ denote the set of unseeded players in bracket $B^k_i$. 
Otherwise, if $x$ is in $B^k_i$, let $U^k_i$ denote the set of unseeded players in bracket $B^k_i$ that are \emph{not} in the same level-$(k + 1)$ section as~$x$.
\item Let $G_k = (T_k', \{B^k_1, \dots, B^k_{|T_k'|}\}, E)$ be the bipartite graph such that there is an edge between $u \in T'_k$ and $B^k_i$ if and only if $u$ loses to a player in $U^k_i$.
\item Find a perfect matching in $G_k$. If no perfect matching exists, fail. Otherwise, if $u \in T_k'$ is matched to $B^k_i$ in the perfect matching, then let $u$ play a player it loses to in $U^k_i$ in the first round.
\end{itemize}
\end{itemize}
Similarly to Case~I, the algorithm runs in polynomial time and, if it succeeds, $x$ is the winner.

To show that the algorithm fails with probability $o(1/n)$, consider $G_k$. Notice that $|U^k_i| \geq n/2^{k + 2}$. Thus, for any $u \in T'_k$, the probability that it loses to at least one player in $U^k_i$ is at least $1 - (1 - p)^{n/2^{k + 2}}$. Applying \Cref{cor:er-matching-customized} with $\delta = (1-p)^{n/2^{k+2}}$ and $m = |T'_k| \geq 2^{k-2}$  ensures that the matching exists with probability at least $1 - 1/n^2$.

Taking a union bound over all $k$, the probability that all of $G_1, \dots, G_{\log n-1}$ admit a perfect matching is at least $1 - \log n / n^2 = 1 - o(1/n)$. 
Moreover, since $p \ge 160\ln n/n \ge 80\ln (n/2)/(n/2)$, \Cref{lem:non-seed-mod} ensures that the first step of the algorithm fails with probability at most $o(1/n)$. 
Hence, the entire algorithm succeeds with probability at least $1 - o(1/n)$, as desired.
\end{proof}

\section{Computational Complexity}
\label{sec:complexity}

In this section, we turn our attention to the complexity of computing a valid winning bracket for our desired winner.
Recall that this problem is NP-complete in the non-seeded setting \citep{AzizGaMa18}.
We show that the intractability continues to hold in the seeded setting.

\begin{theorem}
\label{thm:hardness-constant}
For any constant number of seeds $s$, TFP is NP-complete.
\end{theorem}

\begin{proof}
The problem belongs to NP since a valid winning bracket for our desired winner $x$ can be verified in polynomial time, so we focus on the hardness.

We reduce from TFP in the non-seeded setting.
For ease of presentation, we will first show the reduction for $s = 4$ and later explain how to extend it to any constant value of $s$.
Let $I$ be an instance of non-seeded TFP with a set $V$ of $n$ players, one of whom is our desired winner $x$.
We create sets $V_1$ and $V_2$ of players with the following properties:
\begin{itemize}
\item $|V_1| = n$, and $x_1\in V_1$ beats all other players in $V_1$.
\item $|V_2| = 2n$, and $x_2\in V_2$ beats all other players in $V_2$.
\item All players in $V_1$ beat all players in $V$, with the exception that $x$ beats $x_1$.
\item All players in $V_2$ beat all players in $V$, with the exception that $x$ beats $x_2$.
\item All remaining outcomes can be chosen arbitrarily.
\end{itemize}
The new instance $I'$ consists of $n + n + 2n = 4n$ players.
The top two seeds are $x$ and $x_2$, and the other two seeds are $x_1$ and an arbitrary player $y\in V_2$.
This completes the description of our reduction; an illustration is shown in \Cref{fig:reduction-constant}.
Note that the reduction takes polynomial time.

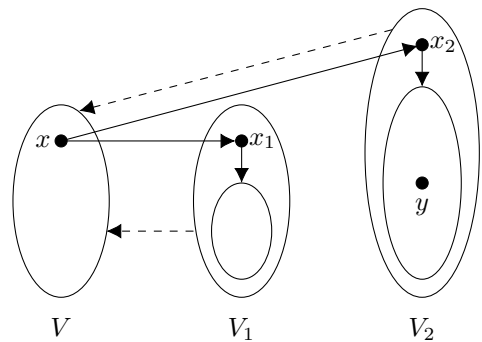
\begin{figure}[!ht]
\centering
\begin{tikzpicture}[scale=0.8]
\draw (3,5) ellipse (0.8cm and 1.6cm);
\draw[fill=black] (3,6) circle  [radius=0.1];
\node at (2.7,6) {$x$};
\node at (3,2.9) {$V$};
\draw (6,5) ellipse (0.8cm and 1.6cm);
\draw[fill=black] (6,6) circle  [radius=0.1];
\node at (6.35,6) {$x_1$};
\draw (6,4.5) ellipse (0.5cm and 0.8cm);
\draw[->] (6,6) to (6,5.3);
\node at (6,2.87) {$V_1$};
\draw (9,5.8) ellipse (0.95cm and 2.4cm);
\draw[fill=black] (9,7.6) circle  [radius=0.1];
\node at (9.35,7.6) {$x_2$};
\draw (9,5.3) ellipse (0.65cm and 1.6cm);
\draw[->] (9,7.6) to (9,6.9);
\draw[fill=black] (9,5.3) circle  [radius=0.1];
\node at (9,4.9) {$y$};
\node at (9,2.87) {$V_2$};
\draw[->] (3,6) to (5.9,6);
\draw[->] (3,6) to (8.93,7.58);
\draw[->,dashed] (5.2,4.5) to (3.75,4.5);
\draw[->,dashed] (8.5,7.85) to (3.3,6.5);
\end{tikzpicture}
\caption{Illustration of the reduction for \Cref{thm:hardness-constant} when $s = 4$}
\label{fig:reduction-constant}
\end{figure}

We claim that $x$ can win in the original instance $I$ if and only if it can win in the new instance $I'$.

($\Rightarrow$) Suppose that $x$ can win in $I$.
To construct a winning bracket for $x$ in $I'$, we use the winning bracket for $x$ in $I$ as one quarter, put the players from $V_1$ in the other quarter in the same half as $x$, and put the players from $V_2$ in the opposite half to $x$ in such a way that $x_2$ and $y$ are in different quarters.
Notice that the constructed bracket satisfies the seed constraints.
In the resulting tournament, $x$ progresses to the semifinals by winning its bracket for $I$, $x_1$ progresses to the semifinals by winning its quarter, and $x_2$ progresses to the final by winning its half.
Hence, $x$ beats $x_1$ in the semifinals and beats $x_2$ in the final, and therefore wins the tournament.

($\Leftarrow$) Suppose that $x$ can win in $I'$, and consider its winning bracket.
We claim that $x$'s quarter must contain exactly the players from $V$.
Indeed, it cannot contain $x_1$ or $x_2$ due to seed constraints, and if it contains players from $(V_1\cup V_2)\setminus\{x_1,x_2\}$, then since all of these players beat all players in $V$, the winner of $x$'s quarter will be from $(V_1\cup V_2)\setminus\{x_1,x_2\}$, and in particular not $x$.
Hence, the bracket of $x$'s quarter is a winning bracket for $x$ in $I$.

This completes the NP-hardness proof for $s = 4$.
To extend it to any constant number of seeds $s=2^t$ (including $s = 2$), we use a similar construction.
Starting with a non-seeded TFP instance $I$ of $n$ players, we create sets $V_1,V_2,\dots,V_t$ of size $n,2n,\dots,2^{t-1}n$.
For each $1\le i\le t$, $V_i$ contains a player $x_i$ who beats all other players in $V_i$, and all players in $V_i$ beat all players in $V$ except that $x$ beats $x_i$.
The new instance $I'$ consists of $2^tn$ players.
The seeds are assigned as follows:
\begin{itemize}
\item The top two seeds are $x$ and $x_t$.
\item The next two seeds are $x_{t-1}$ and a player in $V_t$.
\item The next four seeds are $x_{t-2}$, a player in $V_{t-1}$, and two players in $V_t$.
\item $\dots$
\item The last $2^{t-1}$ seeds are $x_1$, a player in $V_2$, two players in $V_3$, four players in $V_4$, $\dots$, and $2^{t-2}$ players in $V_t$.
\end{itemize}
If $x$ can win in $I$, then we can construct a bracket in $I'$ so that after winning its subtournament from $I$, $x$ beats $x_1,x_2,\dots,x_t$ in the last $t$ rounds to win the tournament for~$I'$.
Conversely, if $x$ can win in $I'$, a similar argument as in the case $s=4$ shows that the section of the bracket containing $x$'s first $\log n$ rounds must contain exactly the players from $V$; this yields a winning bracket for $x$ in $I$.
Finally, since $s$ is constant, the reduction takes polynomial time.
\end{proof}

Next, we prove that the problem remains intractable when the number of seeds is the maximum possible.

\begin{theorem}
\label{thm:hardness-n-2}
For $s = n/2$, TFP is NP-complete.
\end{theorem}

\begin{proof}
As with \Cref{thm:hardness-constant}, we only need to show the NP-hardness.
We again reduce from TFP in the non-seeded setting.
Let $I$ be an instance of non-seeded TFP with a set $V$ of $n$ players, one of whom is our desired winner $x$.
We create a set $V'$ of $n$ additional players who all lose to the $n$ original players.
The outcomes between players in $V'$ can be chosen arbitrarily.
The new instance $I'$ consists of $2n$ players, and the $n$ new players are the $n$ seeds.
It suffices to show that $x$ can win in $I$ if and only if it can win in $I'$.

($\Rightarrow$) Suppose that $x$ can win in $I$.
In the first round for $I'$, we will pair each player from $V$ with a player from $V'$; this ensures that only players from $V$ remain from the second round onward.
We position the players from $V$ so that the bracket from the second round onward is a winning bracket for $x$ in~$I$.
Moreover, we position the players from $V'$ so that the seed constraints are satisfied.
Hence, $x$ wins the resulting tournament for $I'$.

($\Leftarrow$) Suppose that $x$ can win in $I'$, and consider its winning bracket.
By the seed constraints, every first-round match must be between a player from $V$ and one from $V'$; hence, only players from $V$ progress to the second round.
It follows that the bracket from the second round onward is a winning bracket for $x$ in~$I$.
\end{proof}

In the non-seeded setting, the fastest known algorithm for TFP is due to~\citet{KimVa15} and runs in $2^n \cdot n^{O(1)}$ time. 
Not only can the algorithm solve TFP, but it can also count the number of brackets that result in each player becoming the tournament winner. Here we provide an algorithm with similar guarantees in the seeded setting.

\begin{theorem} \label{thm:alg}
For any $n$ and $s$, there exists a $2^n \cdot n^{O(1)}$-time algorithm that outputs, for each player, the number of valid brackets such that the player wins the tournament.
\end{theorem}

To prove this theorem, we need the following lemma of~\citet{KimVa15}, whose proof relies on techniques due to~\citet{BjorklundHuKa07}.
For any positive integer $t$, we use $[t]$ to denote the set $\{1,2,\dots,t\}$.

\begin{lemma}[\citep{KimVa15}] \label{lem:fast-subset-convolution}
Let $i \le \log n$ be a positive integer and $f, g$ be integer-valued functions on subsets of $[n]$ of size $2^{i - 1}$.  Let $h$ be a function on subsets of $[n]$ of size $2^i$ defined by
\begin{align*}
h(S) = \sum_{T \subseteq S \atop |T| = 2^{i - 1}} f(T) \cdot g(S\setminus T).
\end{align*}
If each entry of $f$ and $g$ can be accessed in ${O(1)}$ time, then we can compute $h(S)$ for all $S$ of size $2^i$ in time $2^n \cdot n^{O(1)}$.
\end{lemma}

\begin{proof}[Proof of \Cref{thm:alg}]
Let $[n]$ be the set of players.
For $i = 0, \dots, \log n$ and $j \in [n]$, let us define the function $f_i^j$ on subsets $S\subseteq[n]$ of size $2^i$ as follows:
\begin{itemize}
\item For $i = 0$, $f_0^j(S) = 1$ if $S = \{j\}$, and 0 otherwise.
\item For $i \ge 1$, $f_i^j(S) = 0$ if, for some power-of-two $\ell$ such that $\ell \leq s$ and $\ell \geq 2^{n - i}$, we have $|S \cap [\ell]| \ne \ell / 2^{n - i}$. Otherwise, let
\begin{align*}
f_i^j(S) = \sum_{k \in [n] \atop j \succ k} \sum_{T \subseteq S \atop |T| = 2^{i - 1}} f_{i - 1}^j(T) \cdot f_{i - 1}^k(S\setminus T).
\end{align*}
\end{itemize}
Intuitively, $f_i^j(S)$ captures the number of valid brackets for a subtournament with set of players $S$ (of size $2^i$) such that $j$ wins.
The condition $|S \cap [\ell]| = \ell / 2^{n - i}$ ensures that this subtournament contains exactly one of the top $2^{n-i}$ seeds, two of the top $2^{n-i+1}$ seeds, and so on.
Hence, $f_{\log n}^x([n])$ is exactly the number of valid brackets under which $x$ wins. Therefore, it suffices for us to show how to compute the values of these functions in time $2^n \cdot n^{O(1)}$. Our algorithm is described below. (Note that when we let $f^j_i(S)$ be some number, we actually mean storing each $f^j_i$ as an array and filling in the entry of the array corresponding to $S$.)
\begin{itemize}
\item For $i = 0$, let $f_0^j(\{k\}) = 1$ if $k = j$, and $0$ otherwise.
\item For $i = 1, \dots, \log n$:
\begin{itemize}
\item For $j \in [n]$:
\begin{itemize}
\item Start off with $f_i^j(S) = 0$ for all sets $S$ of size $2^i$.
\item For each $k \in [n]$ such that $j \succ k$:
\begin{itemize}
\item Use \Cref{lem:fast-subset-convolution} to compute $h$ with $f = f_{i - 1}^j$ and $g = f_{i - 1}^k$.
\item For every set $S$ of size $2^i$, increase $f_i^j(S)$ by $h(S)$.
\end{itemize}
\item For each set $S$ of size $2^i$, check whether $|S \cap [\ell]| = \ell / 2^{n - i}$ for every power-of-two $\ell$ such that $\ell \leq s$ and $\ell \geq 2^{n - i}$; if this fails, set $f_i^j(S) = 0$.
\end{itemize}
\end{itemize}
\end{itemize}
The number of pairs $(i,j)$ is $n\log n$, and the iteration of the for-loop corresponding to each pair $(i,j)$ runs in $2^n \cdot n^{O(1)}$ time by \Cref{lem:fast-subset-convolution}.
It follows that the entire algorithm runs in $2^n \cdot n^{O(1)}$ time, as desired.
\end{proof}

\section{Conclusion and Future Work}
\label{sec:conclusion}

In this paper, we have investigated the problem of fixing a knockout tournament in the ubiquitous setting where a subset of the players are designated as seeds.
Our results exhibit both similarities and differences in comparison to the setting without seeds.
On the one hand, the decision problem of whether a certain player can be made a tournament winner remains computationally hard, and manipulation is still feasible in the average case.
On the other hand, a number of structural conditions that guarantee that a player can win in the non-seeded setting cease to do so in the seeded setting.\footnote{Our negative results in \Cref{sec:structural} also apply to the conditions put forth by \citet[Thm.~2.1]{KimSuVa17}, since these conditions generalize both the superking and the ``king of high outdegree'' conditions.}

While the seed constraints that we studied in this work are both common and natural, some real-world tournaments employ more restrictive versions of constraints.
For example, in World Snooker Championships---which are competed among $32$ players, $16$ of whom are seeds---the bracket is set up to ensure that if the top four seeds all make it to the semifinals, the first seed will play against the fourth seed and the second against the third; analogous conditions apply to the quarterfinals as well as the round of $16$.\footnote{See wikipedia.org/wiki/2021\_World\_Snooker\_Championship.}
With these constraints, the only freedom in choosing the bracket is in pairing seeded players with unseeded players in the first round.
Another notable example is the seeding used in Grand Slam tennis tournaments, which involve $128$ players and $32$ seeds.
While the constraints for the semifinals and quarterfinals are the same as those that we studied, additional constraints are imposed in the rounds of $16$ and $32$.
For instance, in the round of $16$, the bracket must be set up so that seeds $9$--$12$ fall in different ``eighths'' from seeds $1$--$4$ \citep[p.~26]{Grandslam21}.
Studying the extent to which manipulation is still possible in such tournaments is an interesting avenue for future research.

\section*{Acknowledgments}

This work was partially supported by an NUS Start-up Grant.
We would like to thank the anonymous reviewers for their valuable comments.

\bibliographystyle{named}
\bibliography{ijcai22}

\end{document}